\documentclass[journal,twoside,web]{ieeecolor}
\usepackage{tikz}
\usepackage{lcsys}
\usepackage{cite}
\usepackage{amsmath,amssymb,amsfonts}
\usepackage{algorithmic}
\usepackage{graphicx}
\usepackage{textcomp}

\usepackage[caption=false]{subfig}

\usepackage{dsfont}
\usepackage{balance}
\usepackage{graphicx}

\usepackage{array}
\usepackage{multirow}
\usepackage{bigdelim}
\usepackage{mdwtab}
\usepackage{physics}
\usepackage{booktabs}

\newcommand\RV[1]{{#1}}

\usepackage{accents}

\newtheorem{theorem}{Theorem}[section]
\newtheorem{lemma}{Lemma}[section]

\newtheorem{corollary}{Corollary}[section]

\def\BibTeX{{\rm B\kern-.05em{\sc i\kern-.025em b}\kern-.08em
    T\kern-.1667em\lower.7ex\hbox{E}\kern-.125emX}}
\begin{document}
\title{ISC-POMDPs: Partially Observed Markov Decision Processes with Initial-State\\ Dependent Costs}
\author{Timothy L. Molloy
\thanks{The author is with the CIICADA Lab, School of Engineering, The Australian National University (ANU), Canberra, ACT 2601, Australia (e-mail: timothy.molloy@anu.edu.au)}}

\pagestyle{empty}

\newcommand\copyrighttext{%
  \footnotesize \textcopyright 2025 IEEE. Personal use of this material is permitted.
  Permission from IEEE must be obtained for all other uses, in any current or future
  media, including reprinting/republishing this material for advertising or promotional
  purposes, creating new collective works, for resale or redistribution to servers or
  lists, or reuse of any copyrighted component of this work in other works.}
\newcommand\copyrightnotice{%
\begin{tikzpicture}[remember picture,overlay]
\node[anchor=south,yshift=5pt] at (current page.south) {\fbox{\parbox{\dimexpr\textwidth-\fboxsep-\fboxrule\relax}{\copyrighttext}}};
\end{tikzpicture}%
}

\maketitle
\thispagestyle{empty}
\copyrightnotice

\begin{abstract}
We introduce a class of partially observed Markov decision processes (POMDPs) with costs that can depend on both the value and (future) uncertainty associated with the initial state.
These Initial-State Cost POMDPs (ISC-POMDPs) enable the specification of objectives relative to \emph{a priori} unknown initial states, which is useful in applications such as robot navigation, controlled sensing, and active perception, that can involve controlling systems to revisit, remain near, or actively infer their initial states.
By developing a recursive Bayesian fixed-point smoother to estimate the initial state that resembles the standard recursive Bayesian filter, we show that ISC-POMDPs can be treated as POMDPs with (potentially) belief-dependent costs.
We demonstrate the utility of ISC-POMDPs, including their ability to select controls that resolve (future) uncertainty about (past) initial states, in simulation.
\end{abstract}


\section{Introduction}
\label{sec:introduction}
\IEEEPARstart{T}{he} initial state of a dynamical system often has important practical significance \cite{Wang2023,Mariottini2011,Xue2022,Shi2025}.
For example, the initial position of a vehicle often corresponds to its owner's residence \cite{Wang2023}; the initial pose of a robot has enabled safe or recoverable navigation, path planning, and mapping \cite{Mariottini2011,Xue2022}; and, the initial configuration of teams of agents can enable recognition of their roles or intent \cite{Shi2025}.
This significance has given rise to partially observed stochastic optimal control problems with objectives tied directly to initial states, such as the problem of controlling a system to hinder inference of its initial state to preserve privacy in networked control systems \cite{Wang2023}, or the problem of controlling a system to improve inference of its initial state for active sensing and perception in target tracking and robotics \cite{Shi2025,Mariottini2011,Xue2022}.
However, a general framework for solving initial-state objective problems is lacking.
We therefore introduce and investigate Initial-State Cost Partially Observed Markov Decision Processes (ISC-POMDPs).

Partially Observed Markov Decision Processes (POMDPs) have costs that depend on the \emph{values} of their (current) partially observed state, and are typically solved by constructing equivalent Markov decision processes (MDPs) with the \emph{belief} (i.e., the conditional distribution of the current state given observed measurements computed via a recursive Bayesian filter) as their ``state'' process.
\RV{POMDPs have been generalized to encompass cost functions (denoted $\rho$ by convention) that are explicit functions of the current belief, not just state value, leading to $\rho$-POMDPs that can also be solved by constructing belief MDPs (cf. \cite{Araya2010,Fehr2018}).}
$\rho$-POMDPs have proved important when controlling state uncertainty is explicitly an objective, such as in active perception, controlled sensing, or privacy-based applications (cf.\ \cite{Araya2010, Fehr2018,Shi2025,Molloy2023} and \cite[Chapter 8]{Krishnamurthy2016}).
Nevertheless, both POMDPs and $\rho$-POMDPs, and their associated belief MDPs, have Markovian dynamics and costs in the sense that the current state (or belief) determines the current cost and evolution of the state (or belief).
Initial-state costs are, however, non-Markovian.

Limited progress has been made in incorporating initial-state costs into ($\rho$-)POMDPs.
Notably, \cite{Shi2025} aimed to minimize the entropy of the conditional distribution of the initial state given (all) measurements at a terminal time by redefining the belief to be this conditional distribution.
However, it is well-known from Bayesian smoothing that the conditional distribution of the initial state does not have a Markovian form computable via a Bayesian filter recursion (cf.\ \cite[Section 4.1.1]{Cappe2005}).
The approach of \cite{Shi2025} therefore does not lead to reformulations of initial-state cost problems as belief MDPs.
More generally, \cite{Belardinelli2022} considered non-Markovian costs in partially observed problems by generalizing the approach of \cite{Bacchus1996} developed for fully observed problems that involves augmenting the underlying state so the dynamics and costs become Markovian.
This approach suffers from the fact that even if the state can be modified so that it is Markovian, it may not have a corresponding Markovian belief.
For example, if the costs only depend on the initial state (or its uncertainty, as in \cite{Shi2025}), then taking this original initial state as a new (static) modified state leads to a trivial constant Markovian state process but the conditional distribution of this modified state (i.e., the original initial state) does not itself have a Markovian form computable via a Bayesian filter recursion (cf.\ \cite[Section 4.1.1]{Cappe2005}).

The key contribution of this paper is the introduction of ISC-POMDPs with costs that can depend on both the value and (future) uncertainty associated with initial states.
We establish that ISC-POMDPs with (arbitrary) initial-state dependent costs admit reformulations and solutions as standard ($\rho$-)POMDPs with augmented state processes consisting of the both initial and current state, and that their associated belief can be computed with a recursive (fixed-point) Bayesian smoother that resembles the standard Bayesian filter.
Surprisingly, state augmentation and recursive Bayesian smoothing have not previously been used to solve initial-state cost problems.

This paper is structured as follows.
We introduce ISC-POMDPs in Section \ref{sec:problem}. 
We reformulate and analyse ISC-POMDPs as ($\rho$-)POMDPs in Section \ref{sec:reforms}.
We provide simulations in Section \ref{sec:results}, and conclusions in Section \ref{sec:conclusion}.

\paragraph*{Notation}
Random variables are denoted by capital letters (e.g., $X$), and their realizations by lower-case letters (e.g., $x$).
The probability mass function (pmf) of $X$ is $p(x)$, the joint pmf of $X$ and $Y$ is $p(x, y)$, and the conditional pmf of $X$ given $Y = y$ is $p(x|y)$ or $p(x | Y = y)$.
The expectation of a function $f$ of $X$ is $E [f(X)]$, and the conditional expectation of $f$ under $p(x|y)$ is $E[f(X) | y]$ or $E[f(X) | Y = y]$.
For a finite set $\mathcal{S}$, the set of all probability distributions (or pmfs) over $\mathcal{S}$ is $\Delta(\mathcal{S})$.

\section{Preliminaries and Problem Formulation}
\label{sec:problem}

We first revisit ($\rho$-)POMDPs and introduce ISC-POMDPs.

\subsection{POMDP and $\rho$-POMDP Preliminaries}

Let $X_k$ for $k \geq 0$ be a discrete-time first-order Markov chain with finite state space $\mathcal{X} \triangleq \{1, 2, \ldots, N_x\}$.
Let the initial state $X_0$ be distributed according to the pmf $\pi_0 \in \Delta(\mathcal{X})$ with components $\pi_0(x_0) \triangleq P(X_0 = x_0)$.\footnote{Note that $\pi_0 \in \Delta(\mathcal{X})$ can be viewed as a $N_x$-dimensional probability vector (a vector with nonnegative components that sum to $1$), and thus $\Delta(\mathcal{X})$ can be viewed as the $(N_x - 1)$-dimensional simplex.}
Let the state $X_k$ evolve according to the state-transition probabilities:
\begin{align}
    \label{eq:stateProcess}
    A^{x,\bar{x}}(u) \triangleq p( X_{k+1} = x | X_k = \bar{x}, U_k = u)
\end{align}
for $k \geq 0$ and $x,\bar{x} \in \mathcal{X}$, with the controls $U_k = u$ belonging to the finite set $\mathcal{U} \triangleq \{1, 2, \ldots, N_u\}$.
The state $X_k$ is (partially) observed through stochastic observations $Y_k$ for $k \geq 1$ taking values in the finite set $\mathcal{Y} \triangleq \{1, 2, \ldots, N_y\}$.
The measurements $Y_k$ are mutually conditionally independent given the states $X_k$ and controls $U_{k-1}$, and distributed according to the measurement probabilities:
\begin{align}
    \label{eq:obsProcess}
    B^x (y, u) \triangleq p( Y_k = y | X_k = x, U_{k-1} = u)
\end{align}
for $k \geq 1$, $x \in \mathcal{X}$, $y \in \mathcal{Y}$, and $u \in \mathcal{U}$.

In (standard infinite-horizon discounted) POMDPs or $\rho$-POMDPs, we may consider the controls $U_k$ to be given by a policy $\mu_\pi : \Delta(\mathcal{X}) \rightarrow \mathcal{U}$ dependent on the \emph{belief} $\pi_k$ about the state $X_k$ given the measurements $Y^k \triangleq \{Y_1, Y_2, \ldots, Y_k\}$ and controls $U^{k-1} \triangleq \{U_0, U_1, \ldots, U_{k-1}\}$ (cf.\ \cite[Section 5.4.1]{Bertsekas2005} and \cite[Chapter 7]{Krishnamurthy2016}).
Specifically, let $U_k = \mu_\pi(\pi_k)$ for $k \geq 0$ where the belief $\pi_{k} \in \Delta(\mathcal{X})$ is a conditional pmf, or $N_x$-dimensional vector, with components $\pi_{k}(x) \triangleq p(X_{k} = x | y^k, u^{k-1})$ satisfying the Bayesian filter recursion
\begin{align}\label{eq:filter}
	\pi_{k+1}(x)
	= \dfrac{B^x(y_{k+1}, u_k) \sum_{\bar{x} \in \mathcal{X}} A^{x,\bar{x}}(u_k) \pi_{k}(\bar{x})}{\sum_{\bar{x}, \tilde{x} \in \mathcal{X}}B^{\tilde{x}}(y_{k+1}, u_k) A^{\tilde{x},\bar{x}}(u_k) \pi_{k}(\bar{x})}
\end{align}
for $x \in \mathcal{X}$ and $k \geq 0$ given $\pi_0$.
We use $\pi_{k+1} = \Pi(\pi_k, u_k, y_{k+1})$ to denote the filter \eqref{eq:filter}, and note that $p(y_{k+1}|\pi_k, u_k)$ is the denominator in \eqref{eq:filter}.
We denote the set of all (deterministic belief) policies $\mu_\pi$ as $\mathcal{P}_\pi$, with $p_{\mu,\pi}$ being the probability law induced by $\mu_\pi \in \mathcal{P}_\pi$ and its corresponding expectation being $E_{\mu,\pi} [\cdot]$.
A POMDP (formulated as a belief MDP) then involves finding a policy that solves
\begin{align}
\label{eq:standardPOMDP}
\begin{aligned}
&\inf_{\mu_\pi \in \mathcal{P}_\pi} & & E_{\mu,\pi} \left[ \sum_{k = 0}^{\infty} \gamma^k C \left( \pi_k, U_{k} \right) \right]\\ 
&\mathrm{s.t.} & & \pi_{k+1} = \Pi(\pi_k, U_k, Y_{k+1}), \quad \pi_0 \in \Delta(\mathcal{X}) \\
& & & Y_{k+1} | \pi_k, U_k \sim p(y_{k+1} | \pi_{k}, u_{k})\\
& & & U_k = \mu_\pi(\pi_k) \in \mathcal{U}
\end{aligned}
\end{align}
given a discount factor $\gamma \in (0,1)$ and a cost function $C : \Delta(\mathcal{X}) \times \mathcal{U} \rightarrow \mathbb{R}$ that is the conditional expectation of an underlying (current) state-control cost function $\kappa : \mathcal{X} \times \mathcal{U} \rightarrow \mathbb{R}$; i.e., $C(\pi_k, u_k) \triangleq E [\kappa(X_k, U_k) | \pi_k, U_k = u_k] = \sum_{x \in \mathcal{X}} \pi_k(x) \kappa(x, u_k)$.
In contrast, a $\rho$-POMDP (formulated as a belief MDP) involves finding a policy solving \cite{Araya2010}
\begin{align}
\label{eq:rhoPOMDP}
\begin{aligned}
&\inf_{\mu_\pi \in \mathcal{P}_\pi} & & E_{\mu, \pi} \left[ \sum_{k = 0}^{\infty} \gamma^k \rho \left( \pi_k, U_k \right) \right]\\ 
&\mathrm{s.t.} & & \pi_{k+1} = \Pi(\pi_k, U_k, Y_{k+1}), \quad \pi_0 \in \Delta(\mathcal{X}) \\
& & & Y_{k+1} | \pi_k, U_k \sim p(y_{k+1} | \pi_{k}, u_{k})\\
& & & U_k = \mu_\pi(\pi_k) \in \mathcal{U}
\end{aligned}
\end{align}
given a discount factor $\gamma \in (0,1)$ and an \emph{arbitrary} belief-dependent cost function $\rho : \Delta(\mathcal{X}) \times \mathcal{U} \rightarrow \mathbb{R}$.
$\rho$-POMDPs \eqref{eq:rhoPOMDP} generalize POMDPs since the cost function $\rho$ can be any (potentially nonlinear) function of the belief $\pi_k$ whilst the POMDP cost function $C$ is limited to the linear form implied by conditional expectation.
We propose ISC-POMDPs as an extension of ($\rho$-)POMDPs with costs that can depend on both the value and uncertainty associated with the initial state $X_0$.

\subsection{ISC-POMDPs}

To introduce ISC-POMDPs, let us consider the possibility of the controls $U_k$ for $k \geq 0$ being given by a policy $\mu$ that is a (deterministic) function of the measurements and controls $(Y^k, U^{k-1})$ directly, namely, $U_k = \mu(Y^k, U^{k-1})$.
Let the set of all such policies be $\mathcal{P}$, and let the probability law induced by a policy $\mu \in \mathcal{P}$ be $p_{\mu}$ with corresponding expectation $E_{\mu} [\cdot]$.
Let us also define the (joint posterior) conditional pmf of the state $X_k$ and the initial state $X_0$ given the information $(y^k, u^{k-1})$ available at time $k \geq 0$ as $\xi_k \in \Delta(\mathcal{X} \times \mathcal{X})$, where
\begin{align}
	\label{eq:firstInitialBelief}
	\xi_k(x_0, x_k)
	&\triangleq p(X_0 = x_0, X_k = x_k | y^k, u^{k-1})
\end{align}
for $x_0, x_k \in \mathcal{X}$.
We introduce an ISC-POMDP as the problem of finding a policy that solves
\begin{align}
    \label{eq:iscpomdp}
    \begin{aligned}
    & \inf_{\mu \in \mathcal{P}} & &\hspace{-0.3cm} E_\mu \left[ \sum_{k = 0}^\infty \gamma^k \left[ c(X_0, X_k, U_k) + \psi(\xi_k, U_k) \right] \right]\\
    &\mathrm{s.t.} & &\hspace{-0.3cm}   X_{k+1} | X_k, U_k \sim A^{x_{k+1},x_k}(u_k), X_0 \sim \pi_0 \in \Delta(\mathcal{X})\\
    & & &\hspace{-0.3cm} Y_{k+1} | X_{k+1}, U_k \sim B^{x_{k+1}}(y_{k+1}, u_k),\\
    & & &\hspace{-0.3cm} U_k = \mu(Y^k, U^{k-1}) \in \mathcal{U}
    \end{aligned}
\end{align}
for a given discount factor $\gamma \in (0,1)$ where $c : \mathcal{X} \times \mathcal{X} \times \mathcal{U} \rightarrow \mathbb{R}$ is an arbitrary cost function dependent on the \emph{values} of the current state $X_k$, initial state $X_0$, and controls $U_k$, and where $\psi : \Delta(\mathcal{X} \times \mathcal{X}) \times \mathcal{U} \rightarrow \mathbb{R}$ is an arbitrary function of the (joint) \emph{posterior pmf} $\xi_k$ and the controls $U_k$.

ISC-POMDPs \eqref{eq:iscpomdp} generalize POMDPs \eqref{eq:standardPOMDP} by introducing costs $c(X_0, X_k, U_k)$ that can depend on the value of the initial state $X_0$, not just on the value of the current state $X_k$.
This dependence enables objectives to specified with respect to the initial state; for example, the cost $c(x_0, x, u) = |x_0 - x|$ for $x_0, x \in \mathcal{X}$ and $u \in \mathcal{U}$ specifies the objective of keeping (future) states $X_k$ close to the (potentially \emph{a priori} unknown) initial state $X_0$.
ISC-POMDPs \eqref{eq:iscpomdp} also generalize $\rho$-POMDPs \eqref{eq:rhoPOMDP} by introducing costs $\psi(\xi_k, U_k)$ that can depend on the joint posterior $\xi_k$ of the initial $X_0$ and current state $X_k$ at any time $k \geq 0$ (and hence also their marginals).
This generalization enables the optimization of uncertainty measures associated with the initial and current states, and exploits that (future) measurements $Y_k$ for $k \geq 1$ can contain information about the (past) initial state $X_0$.
For example, Bayesian (fixed-point smoother) estimates of the initial state at future times $k > 0$ can be improved by solving the ISC-POMDP \eqref{eq:iscpomdp} with the initial-state entropy cost $	\psi(\xi_k, U_k)
	= H(X_0 | y^k, u^{k-1})
	\triangleq - \sum_{x_0 \in \mathcal{X}} p(x_0 | y^k, u^{k-1}) \log p(x_0 | y^k, u^{k-1})$ where $p(x_0 | y^k, u^{k-1}) = \sum_{x_k \in \mathcal{X}} \xi_k(x_0, x_k)$.

We propose solving ISC-POMDPs \eqref{eq:iscpomdp} by reformulating them as ($\rho$-)POMDPs with augmented state processes consisting of both the original initial state $X_0$ and the original current state $X_k$.
We shall show that this choice of augmented state leads to an associated (augmented) belief, equivalent to the joint posterior pmf $p(x_0, x_k | y^k, u^{k-1})$, that is Markovian and given by a Bayesian filter recursion, thus enabling the solution of ISC-POMDPs using standard ($\rho$-)POMDP techniques (cf.\ \cite{Kurniawati2008,Araya2010,Fehr2018,Zheng2023}).
We note that our choice of this augmented state is \emph{necessary} to reformulate ISC-POMDPs as ($\rho$-)POMDPs since insight from recursive Bayesian smoothing implies that only the joint posterior pmf $p(x_0, x_k | y^k, u^{k-1})$ has a recursive form whilst the marginal posterior pmf $p(x_0 | y^k, u^{k-1})$ does not (cf.\ \cite[Section 4.1.1]{Cappe2005}).
Interestingly, this insight implies that the joint posterior pmf $p(x_0, x_k | y^k, u^{k-1})$ must be used as the belief for ISC-POMDPs, even when their costs only depend on the initial state $X_0$ or its posterior pmf $p(x_0 | y^k, u^{k-1})$, as in the case of the initial-state entropy $H(X_0 | y^k, u^{k-1})$. 


\section{Reformulation and Solution of ISC-POMDPs}
\label{sec:reforms}

In this section, we reformulate ISC-POMDPs as ($\rho$-)POMDPs with an augmented state and a Markovian belief.

\subsection{Augmented State and Belief Construction}

Let us introduce the augmented state
\begin{align}
	\label{eq:augmentedState}
    S_k
    &\triangleq X_0 + N_x(X_k - 1) \in \mathcal{S}
\end{align}
for $k \geq 0$, with corresponding augmented state space $\mathcal{S} \triangleq \{1, 2, \ldots, N_s\}$ where $N_s \triangleq N_x \times N_x$.
\RV{The augmented state $S_k$ provides an invertible representation of the pair $(X_0,X_k)$ in the sense that given $(X_0,X_k)$, we can compute $S_k$ via \eqref{eq:augmentedState}, and given $S_k$ we can compute $(X_0, X_k)$ via
$
	X_0
	= S_k - N_x \left\lfloor (S_k -1)/N_x \right\rfloor
	\text{ and }
	X_k
	= (S_k - X_0)/N_x + 1
$ where $\lfloor \cdot \rfloor$ denotes the floor function.
Let $\mathcal{L}(x_0,x_k) \triangleq x_0 + N_x(x_k - 1) \in \mathcal{S}$ for $x_0,x_k \in \mathcal{X}$ be the mapping implied by \eqref{eq:augmentedState}.}\footnote{\RV{The mapping $\mathcal{L}$ is analogous to those used to construct linear, vectorized, or ``flattened'' indices of matrices (or tensors), with $S_k$ being the linear index of the pair $(X_0,X_k)$. More generally, the augmented state could be constructed with any invertible (i.e., bijective) mapping $\mathcal{L} : \mathcal{X} \times \mathcal{X} \rightarrow \mathcal{S}$.}}

To derive the probabilistic structure of the augmented states $S_k$, let the pmf describing the initial augmented state $S_0$ be $\xi_0 \in \Delta(\mathcal{S})$ with $\xi_0(s) \triangleq p(S_0 = s)$ for $s \in \mathcal{S}$.
Similarly, let the transition probabilities for the augmented states be $\overline{A}^{s,\bar{s}}(u)
	\triangleq p(S_{k+1} = s | S_k = \bar{s}, U_k = u)$ for $k \geq 0$, $s,\bar{s} \in \mathcal{S}$, and $u \in \mathcal{U}$.
Finally, let the measurement probabilities for the augmented states be $\overline{B}^s(y,u) \triangleq p(Y_k = y | S_k = s, U_{k-1} = u)$ for $k \geq 1$, $s \in \mathcal{S}$, $y \in \mathcal{Y}$ and $u \in \mathcal{U}$.
These probabilities are developed in the following lemma.

\begin{lemma}
	\label{lemma:augmentedProbs}
	Under the constraints in the ISC-POMDP \eqref{eq:iscpomdp}, the initial augmented-state probabilities satisfy
\begin{align} \label{eq:augmentedInitial}
    \xi_0(s)
    = \begin{cases}
        \pi_0(x) & \text{if } s = \mathcal{L}(x,x),\\
        0 & \text{otherwise}
    \end{cases}
\end{align}
for $s \in \mathcal{S}$ and $x \in \mathcal{X}$; the augmented state-transition probabilities satisfy
\begin{align}\label{eq:augmentedTransitions}
	\overline{A}^{s,\bar{s}}(u)
	&= \begin{cases}
		A^{x,\bar{x}} (u) & \text{if } s = \mathcal{L}(x_0, x) \text{ \& } \bar{s} = \mathcal{L}(x_0, \bar{x}),\\
		0 & \text{otherwise}
	\end{cases}
\end{align}
for $s, \bar{s} \in \mathcal{S}$, $x,\bar{x}, x_0 \in \mathcal{X}$, and $u \in \mathcal{U}$; and, the augmented-state measurement probabilities satisfy
\begin{align}\label{eq:augmentedMeasurements}
	\overline{B}^s(y,u)
	&= B^x(y,u)
\end{align}
for $s = \mathcal{L}(x_0, x) \in \mathcal{S}$, $y \in \mathcal{Y}$, and $u \in \mathcal{U}$ where $x_0, x \in \mathcal{X}$.
\end{lemma}
\begin{proof}
By definition $S_0 = \mathcal{L}(X_0,X_0)$, and thus for any $s = \mathcal{L}(x, \bar{x}) \in \mathcal{S}$ where $x,\bar{x} \in \mathcal{X}$ we have that
\begin{align*}
    p(S_0 = s) = p(X_0 = x, X_0 = \bar{x})
    = \begin{cases}  
    	\pi_0(x) & \text{if } x = \bar{x}\\
        0 & \text{otherwise,}
    \end{cases}
 \end{align*}
 proving \eqref{eq:augmentedInitial}.
Consider any $k \geq 0$, $u \in \mathcal{U}$, $s = \mathcal{L}(x_0, x) \in \mathcal{S}$ and $\bar{s} = \mathcal{L}(\bar{x}_0, \bar{x}) \in \mathcal{S}$ where $x,\bar{x}, x_0, \bar{x}_0 \in \mathcal{X}$, then
\begin{align*}
	&p(S_{k+1} = s | S_k = \bar{s}, U_k = u)\\
	&= p(X_{k+1} = x, X_0 = x_0 | X_k = \bar{x}, X_0 = \bar{x}_0, U_k = u)\\
	&=	\begin{cases}
		p(X_{k+1} = x | X_k = \bar{x}, X_0 = x_0, U_k = u) & \text{if } x_0 = \bar{x}_0 \\
		0 & \text{otherwise,}
	\end{cases}
\end{align*}
proving \eqref{eq:augmentedTransitions} since $p(X_{k+1} = x | X_k = \bar{x}, X_0 = x_0, U_k = u) = p(X_{k+1} = x | X_k = \bar{x}, U_k = u) = \overline{A}^{s,\bar{s}}(u)$ due to $X_{k+1}$ and $X_0$ being conditionally independent given $X_k$ and $U_k$.
Finally, consider any $k \geq 1$, $s = \mathcal{L}(x_0, x_0) \in \mathcal{S}$, $y \in \mathcal{Y}$, and $u \in \mathcal{U}$ where $x_0, x \in \mathcal{X}$, then $p(Y_k = y | S_k = s, U_{k-1} = u)
		= p(Y_k = y | X_k = x, U_{k-1} = u)$ since $Y_k$ and $X_0$ are conditionally independent given $X_k$ and $U_{k-1}$, proving \eqref{eq:augmentedMeasurements}.
The proof is complete.
\end{proof}

The conditional pmf $\xi_k$ in \eqref{eq:firstInitialBelief} corresponds to the \emph{(augmented) belief} for the augmented state $S_k$.
In a slight abuse of notation, we shall therefore denote the conditional pmf $\xi_k$ introduced in \eqref{eq:firstInitialBelief} as $\xi_k \in \Delta(\mathcal{S})$ with $\xi_k(s) \triangleq p(S_k = s | y^k, u^{k-1})$ for $s \in \mathcal{S}$ and $k \geq 0$.
The augmented state in \eqref{eq:augmentedState} enables us to treat $\xi_k$ as an $N_s$-dimensional probability vector and $\Delta(\mathcal{S})$ as the $(N_s - 1)$-dimensional simplex.
The following lemma establishes that the augmented belief $\xi_k$ evolves via a simple recursion resembling the Bayesian filter \eqref{eq:filter}.

\begin{lemma}
	\label{lemma:augmentedFilter}
	Under the constraints in the ISC-POMDP \eqref{eq:iscpomdp}, the augmented belief $\xi_k$ evolves via the recursion
\begin{align}
	\label{eq:augmentedFilter}
	\xi_{k+1}(s)
	&= \dfrac{\overline{B}^{s}(y_{k+1}, u_k) \sum_{\bar{s} \in \mathcal{S}} \overline{A}^{s,\bar{s}}(u_k) \xi_{k}(\bar{s})}{\sum_{\bar{s}, \tilde{s} \in \mathcal{S}}\overline{B}^{\tilde{s}}(y_{k+1}, u_k) \overline{A}^{\tilde{s},\bar{s}}(u_k) \xi_{k}(\bar{s})}
\end{align}
for $s \in \mathcal{S}$ and $k \geq 0$ given the initial augmented belief $\xi_0 \in \Delta(\mathcal{S})$, the augmented state-transition probabilities \eqref{eq:augmentedTransitions} and the augmented measurement probabilities \eqref{eq:augmentedMeasurements}.
Furthermore,
\begin{align}
	\label{eq:augmentedNormalization}
	p(y_{k+1} | \xi_k, u_k)
	&= \sum_{\bar{s}, \tilde{s} \in \mathcal{S}}\overline{B}^{\tilde{s}}(y_{k+1}, u_k) \overline{A}^{\tilde{s},\bar{s}}(u_k) \xi_{k}(\bar{s})
\end{align}
for $y_{k+1} \in \mathcal{Y}$, $\xi_k \in \Delta(\mathcal{S})$, and $u_k \in \mathcal{U}$.
\end{lemma}
\begin{proof}
	Consider any $k \geq 0$.
	To establish \eqref{eq:augmentedFilter}, we show 
\begin{align}\label{eq:firstEq}
    p(s_{k+1} | y^{k+1}, u^{k})
    &= \dfrac{p(y_{k+1} | s_{k+1}, u_k) p(s_{k+1} | y^k, u^k)}{p(y_{k+1} | y^k, u^{k})}
\end{align}
where
\begin{align}\label{eq:secondEq}
	p(s_{k+1} | y^k, u^k)
	&= \sum_{s_k \in \mathcal{S}} p(s_{k+1} | s_k, u_k) p(s_k | y^k, u^{k-1})
\end{align}
since \eqref{eq:augmentedFilter} follows from \eqref{eq:firstEq} and \eqref{eq:secondEq} via the definitions of the augmented state-transition probabilities \eqref{eq:augmentedTransitions}, augmented likelihoods \eqref{eq:augmentedMeasurements}, and $\xi_{k+1}(s_{k+1}) = p(s_{k+1} | y^{k+1}, u^{k})$.

To see that \eqref{eq:firstEq} holds, note that Bayes' rule gives
\begin{align*}
	p(s_{k+1} | y^{k+1}, u^{k})
	&= \dfrac{p(y_{k+1} | s_{k+1}, y^k, u^k) p(s_{k+1} | y^k, u^k)}{p(y_{k+1} | y^k, u^{k})}\\
	&= \dfrac{p(y_{k+1} | s_{k+1}, u_k) p(s_{k+1} | y^k, u^k)}{p(y_{k+1} | y^k, u^{k})}
\end{align*}
where the last line holds because $Y_{k+1}$ and $(Y^k, U^k)$ are conditionally independent given $X_{k+1}$ with $U_k = \mu(Y^k,U^{k-1})$, and thus $Y_{k+1}$ and $(Y^k, U^k)$ are conditionally independent given $S_{k+1} = \mathcal{L}(X_0, X_{k+1})$.
Next, \eqref{eq:secondEq} holds since
\begin{align*}
	p(s_{k+1} | y^k, u^k)
	&= \sum_{s_k \in \mathcal{S}} p(s_{k+1} | s_k, y^k, u^k) p(s_k | y^k, u^k)\\
	&= \sum_{s_k \in \mathcal{S}} p(s_{k+1} | s_k, u_k) p(s_k | y^k, u^{k-1})
\end{align*}
where: i) $X_{k+1}$ and $(Y^k, U^{k-1})$ are conditionally independent given $X_k$, and thus $S_{k+1} = \mathcal{L}(X_0, X_{k+1})$ and $(Y^k, U^{k-1})$ are conditionally independent given $S_k = \mathcal{L}(X_0, X_k)$; and ii) $S_k$ and $U_k$ are conditionally independent given $(Y^k, U^{k-1})$ since $U_k = \mu(Y^k, U^{k-1})$.
Thus \eqref{eq:firstEq} and \eqref{eq:secondEq} hold, implying \eqref{eq:augmentedFilter}, with Bayes' rule and the law of total probability giving \eqref{eq:augmentedNormalization}.
The proof is complete.
\end{proof}

We shall use the shorthand $\xi_{k+1} = \Xi(\xi_k, u_k, y_{k+1})$ to denote the recursion in \eqref{eq:augmentedFilter} since it (surprisingly) resembles the Bayesian filter \eqref{eq:filter} but for the augmented state $S_k$ with measurements $Y_k$ and controls $U_k$.
Interestingly, \eqref{eq:augmentedFilter} is equivalently a recursive fixed-point Bayesian smoother for the initial state $X_0$, with similar recursive smoothers explored in \cite[Section 4.1.1]{Cappe2005} and references therein, but not exploited in POMDPs (with controls).
We next show that \eqref{eq:augmentedFilter} enables the reformulation of ISC-POMDPs \eqref{eq:iscpomdp} as ($\rho$-)POMDPs.

\subsection{Augmented ($\rho$-)POMDP Reformulation}

Our main result reformulating ISC-POMDPs follows.

\begin{theorem}
\label{theorem:MDP}
Consider the ISC-POMDP \eqref{eq:iscpomdp}.
Define the augmented-belief cost function
\begin{align}
    \label{eq:augmentedBeliefCost}
   	&\overline{\rho}(\xi, u)
    \triangleq \psi(\xi,u) + \sum_{s \in \mathcal{S}} \xi(s) c(s, u)
\end{align}
for $\xi \in \Delta(\mathcal{S})$ and $u \in \mathcal{U}$, where in a slight abuse of notation we define the augmented cost function $c(s,u) \triangleq c(x_0, x, u)$ for $u \in \mathcal{U}$ and $s = \mathcal{L}(x_0, x) \in \mathcal{S}$.
Then the ISC-POMDP \eqref{eq:iscpomdp} is equivalent to the augmented-belief $\rho$-POMDP:
\begin{align}
\label{eq:iscMDP}
\begin{aligned}
&\inf_{\overline{\mu} \in \overline{\mathcal{P}}} & & E_{\overline{\mu}} \left[ \sum_{k = 0}^{\infty} \gamma^k \overline{\rho}(\xi_k, U_k) \right]\\ 
&\mathrm{s.t.} & &  \xi_{k+1} = \Xi\left( \xi_{k}, U_{k}, Y_{k+1} \right), \quad \xi_0 \in \Delta(\mathcal{S}) \\
& & & Y_{k+1} | \xi_k, U_k \sim p(y_{k+1} | \xi_k, u_k)\\
& & & U_k = \overline{\mu}(\xi_k) \in \mathcal{U}
\end{aligned}
\end{align}
where the optimization is over (deterministic) policies $\overline{\mu} : \Delta(\mathcal{S}) \rightarrow \mathcal{U}$ that are functions of the augmented belief $\xi$, with $\overline{\mathcal{P}}$ being the set of all such policies.
\end{theorem}
\begin{proof}
For $\mu \in \mathcal{P}$, the cost functional in \eqref{eq:iscpomdp} satisfies
\begin{align*}
	&E_\mu \left[ \sum_{k = 0}^\infty \gamma^k \left[ c(X_0, X_k, U_k) + \psi(\xi_k, U_k) \right] \right]\\
	&\;= E_\mu \left[ \sum_{k = 0}^\infty \gamma^k E [ c(X_0, X_k, U_k) + \psi(\xi_k, U_k) | Y^k, U^{k}] \right]\\
	&\;= E_\mu \left[ \sum_{k = 0}^\infty \gamma^k \overline{\rho}(\xi_k, U_k) \right]
\end{align*}
 where the first equality is due to the tower property of expectations; and, the second equality holds due to the (augmented) belief $\xi_k$ being a sufficient statistic for $(Y^k,U^{k-1})$ since $\xi_k$ is a function of $(Y^k,U^{k-1})$ and so
 \begin{align*}
 	&E [ c(X_0, X_k, U_k) + \psi(\xi_k, U_k) | Y^k, U^{k}]\\
 	&\quad= \psi(\xi_k, U_k) + \sum_{s \in \mathcal{X}} \xi_k(s) c(s, U_k) = \overline{\rho}(\xi_k,U_k).
 \end{align*}
 Since $\xi_k$ is a controlled Markov process via Lemma \ref{lemma:augmentedFilter}, POMDP (or belief MDP) results imply that this expectation can be minimized over $\mu \in \mathcal{P}$ under the constraints in \eqref{eq:iscpomdp} by (deterministic) functions $\overline{\mu} \in \overline{\mathcal{P}}$ of $\xi_k$ (cf.\ \cite[Section 5.4.1]{Bertsekas2005} and \cite[Theorem 6.2.2]{Krishnamurthy2016}). 
The proof is complete.
\end{proof}

A special case of Theorem \ref{theorem:MDP} is that when there is no (explicit) belief-dependent cost $\psi$, ISC-POMDPs \eqref{eq:iscpomdp} reduce to POMDPs with states $\mathcal{S}$, measurements $\mathcal{Y}$, controls $\mathcal{U}$, and transition and observations probabilities \eqref{eq:augmentedTransitions} and \eqref{eq:augmentedMeasurements}.
\begin{corollary}
	\label{corollary:augmentedPOMDP}
	If $\psi(\xi,u) = 0$ for $\xi \in \Delta(\mathcal{S})$ and $u \in \mathcal{U}$ in \eqref{eq:iscpomdp}, then \eqref{eq:iscpomdp} is equivalent to the (augmented-state) POMDP
\begin{align}
\label{eq:augmentedPOMDP}
\begin{aligned}
&\inf_{\overline{\mu} \in \overline{\mathcal{P}}} & & E_{\overline{\mu}} \left[ \sum_{k = 0}^{\infty} \gamma^k \overline{C} \left( \xi_k, U_{k} \right) \right]\\ 
&\mathrm{s.t.} & & \xi_{k+1} = \Xi(\xi_k, U_k, Y_{k+1}), \quad \xi_0 \in \Delta(\mathcal{S}) \\
& & & Y_{k+1} | \xi_k, U_k \sim p(y_{k+1} | \xi_{k}, u_{k})\\
& & & U_k = \overline{\mu}(\xi_k) \in \mathcal{U}
\end{aligned}
\end{align}
where $\overline{C}(\xi, u) \triangleq \sum_{s \in \mathcal{S}} \xi(s) c(s,u)$ for $\xi \in \Delta(\mathcal{S})$ and $u \in \mathcal{U}$.
\end{corollary}

Corollary \ref{corollary:augmentedPOMDP} implies that all techniques for solving or analyzing POMDPs of the form \eqref{eq:standardPOMDP} apply directly to ISC-POMDPs \eqref{eq:iscpomdp} that do not have a belief-dependent cost function $\psi$.
Theorem \ref{theorem:MDP}, more generally, implies that an optimal policy $\mu^* : \Delta(\mathcal{S}) \rightarrow \mathcal{U}$ and value function $V : \Delta(\mathcal{S}) \rightarrow \mathbb{R}$ solving an ISC-POMDP \eqref{eq:iscpomdp} with arbitrary belief-dependent cost function $\psi$ can be found via Bellman's equation
\begin{align}
    \label{eq:augmentedValueFunction}
    V(\xi)
    = \min_{u \in \mathcal{U}} \left\{ \overline{\rho}(\xi, u) + \gamma E \left[ V(\Xi(\xi, u, Y)) | \xi, u \right] \right\}
\end{align}
for all $\xi \in \Delta(\mathcal{S})$, with $\mu^*(\xi)$ being a minimizing argument in \eqref{eq:augmentedValueFunction} (cf.\ \cite[Theorem 6.2.2]{Krishnamurthy2016}).
We next discuss structural results useful for finding solutions to ISC-POMDPs via \eqref{eq:augmentedValueFunction}.

\subsection{Structural Results and Approximate Solutions}

The structural result of foremost utility is that the value function $V$ is concave when $\psi$ is concave (or constant) in $\xi$.

\begin{theorem}
\label{theorem:concave}
Consider the ISC-POMDP \eqref{eq:iscpomdp} reformulated as the $\rho$-POMDP \eqref{eq:iscMDP}.
If $\psi(\cdot, u)$ is concave in $\xi \in \Delta(\mathcal{S})$ for $u \in \mathcal{U}$, then $\overline{\rho}(\cdot,u)$ is concave in $\xi \in \Delta(\mathcal{S})$ for $u \in \mathcal{U}$ and the value function $V$ given by \eqref{eq:augmentedValueFunction} is concave in $\xi \in \Delta(\mathcal{S})$.
\end{theorem}
\begin{proof}	
	Given \eqref{eq:augmentedBeliefCost}, that $\overline{\rho}(\cdot,u)$ is concave in $\xi \in \Delta(\mathcal{S})$ for $u \in \mathcal{U}$ when $\psi(\cdot,u)$ is concave in $\xi \in \Delta(\mathcal{S})$ for $u \in \mathcal{U}$ holds since it is the sum of concave functions.
	With this concavity, the theorem follows via \cite[Theorem 3.1]{Araya2010}.
\end{proof}

Theorem \ref{theorem:concave} implies that the reformulation in \eqref{eq:iscMDP} of ISC-POMDPs \eqref{eq:iscpomdp} with concave belief-dependent cost functions $\psi (\cdot, u)$ is a $\rho$-POMDP amenable to approximate solution via the approach developed in \cite{Araya2010}.
Indeed, following \cite{Araya2010} and using the concavity of $\overline{\rho}(\cdot,u)$ established in Theorem \ref{theorem:concave}, we can first construct a piecewise-linear concave (PWLC) approximation of $\overline{\rho}(\cdot,u)$ for $u \in \mathcal{U}$, before then using standard POMDP solvers to compute PWLC approximations of the value function $V$ (see \cite[Section 4]{Araya2010} for details).
The approximation errors are bounded if $\overline{\rho}$ satisfies the H{\"o}lder-continuity conditions of \cite[Theorem 4.3]{Araya2010}, and can, in principle, be made arbitrarily small (see \cite[Section 4.2]{Araya2010} for details).
Many popular uncertainty costs are concave and satisfy the conditions of \cite[Theorem 4.3]{Araya2010}.
For example, the initial-state entropy $H(X_0 | y^k, u^{k-1})$ is concave in $\xi_k$, and entropy functionals satisfy the conditions of \cite[Theorem 4.3]{Araya2010} (cf.\ \cite{Araya2010}).
However, if $\psi$ is not concave but is Lipschitz in $\xi$, then recent Lipschitz-based approximations can be used (see \cite{Fehr2018,Demirci2024}).

\RV{Finally, the reformulations of ISC-POMDPs in Theorem \ref{theorem:MDP} and Corollary \ref{corollary:augmentedPOMDP} have state and belief spaces $\mathcal{S}$ and $\Delta(\mathcal{S})$ that scale quadratically with the state space $\mathcal{X}$.
However, they enable the solution of ISC-POMDPs with state-of-the-art offline and online POMDP solvers capable of handling very large state spaces (cf.\ \cite{Fehr2018,Krishnamurthy2016, Kurniawati2008,Zheng2023}).
In contrast, approaches tailored to specific initial-state costs (such as that of \cite{Shi2025} for the entropy $H(X_0 | y^k, u^{k-1})$) have computational and memory requirements that must be carefully managed via parameters such as memory length and number of samples.}

\section{Simulation Experiment}
\label{sec:results}

We now illustrate using ISC-POMDPs to optimize costs defined with respect to an \emph{a priori} unknown initial state $X_0$.

Consider an agent moving in the grid shown in Fig.\ \ref{fig:grid1} that seeks to move to the corner closest to its initial position.
Each cell in the grid is a state in the state space $\mathcal{X} = \{1, \ldots, 16\}$ (enumerated top-to-bottom, left-to-right).
The agent's initial position is \emph{a priori} unknown and distributed uniformly (i.e., $\pi_0$ is uniform).
The agent's controls $\mathcal{U} = \{1, \ldots, 5\}$ correspond to moving one cell in each of the four compass directions, or staying still.
The controls fail (and the agent remains still) with probability $0.2$, and the walls (bold black lines in Fig.\ \ref{fig:grid1}) block movement, with the agent staying still if it attempts to move into them.
The agent receives measurements $\mathcal{Y} = \{1,\ldots,16\}$ corresponding to whether or not a wall is immediately adjacent to its current cell in each of the four compass directions.
A wall is detected when it is present (resp.\ not present) with probability $0.8$ (resp.\ $0.2$).

\begin{figure}[t!]
     \centering
     \subfloat[]{\includegraphics[width=0.44\columnwidth]{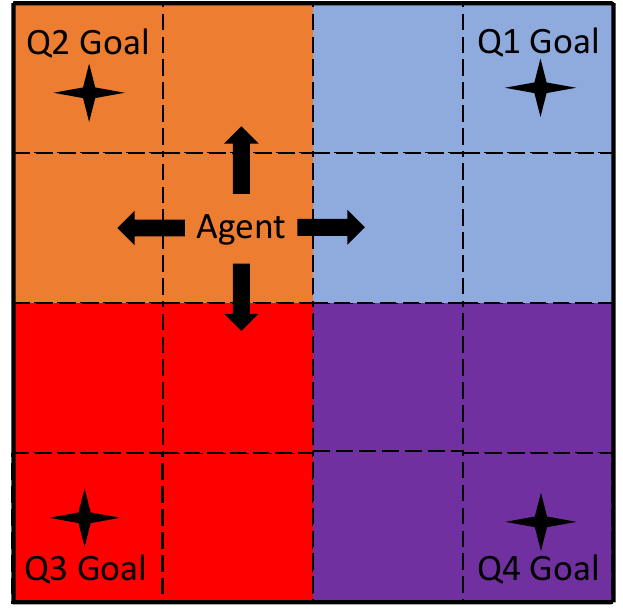}\label{fig:grid1}}
     \hfil
     \subfloat[]{\includegraphics[width=0.44\columnwidth]{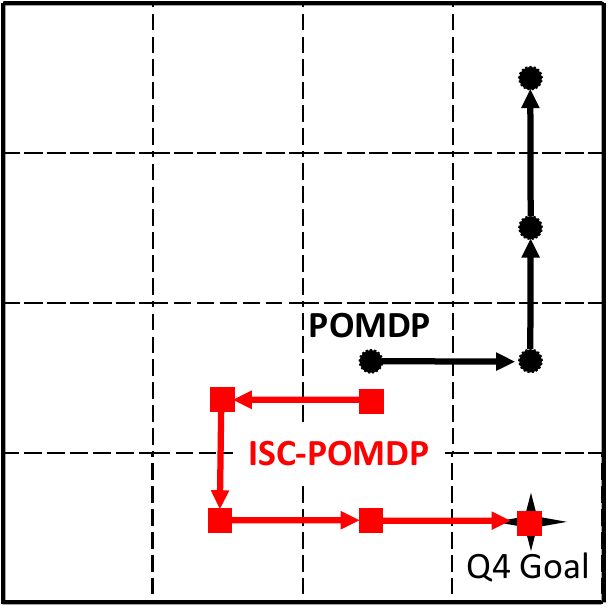}\label{fig:grid2}}
    \caption{Simulation Experiment: (a) Agent must move to goal in corner of quadrant of initial state $X_0$ (agent shown must move to Q2 Goal). (b) Realizations with POMDP moving to corner closet to location $X_k$ for $k = 2$ but ISC-POMDP taking steps to estimate $X_0$ then moving to correct goal (Q4 Goal).}
    \label{fig:grid}
\end{figure}

We encode the agent's problem as an ISC-POMDP \eqref{eq:iscpomdp} \RV{with no belief-dependent cost function (i.e., with $\psi(\xi,u) = 0$ for $\xi \in \Delta(\mathcal{S})$ and $u \in \mathcal{U}$) and with initial-state cost function
\begin{align}
	\label{eq:iscCost}
	c(x_0, x, u) 
	&= \begin{cases}
	 	\mathds{1}(x \neq 1) & \text{if } x_0 \in \{1, 2, 5, 6\} \\
	 	\mathds{1}(x \neq 4) & \text{if } x_0 \in \{3, 4, 7, 8\} \\
	 	\mathds{1}(x \neq 13) & \text{if } x_0 \in \{9, 10, 13, 14\} \\
	 	\mathds{1}(x \neq 16) & \text{if } x_0 \in \{11, 12, 15, 16\} \\
 	   \end{cases}
\end{align}
}%
for $x \in \mathcal{X}$ and $u \in \mathcal{U}$, \RV{where $\mathds{1}(\cdot)$ denotes the indicator function.}
Encoding the agent's objective of moving to the corner closest to its initial position is not directly possible using a (standard) POMDP \eqref{eq:standardPOMDP} since they are limited to current-state dependent costs.
For the purpose of comparison, we therefore encode an approximation of the agent's objective within a POMDP \eqref{eq:standardPOMDP} with cost \RV{$\kappa(x, u) = \mathds{1}(x \not\in \{1,4,13,16\})$} for $u \in \mathcal{U}$.
This cost is only an approximation since it encourages the agent to move to the corner closest to its current location $X_k$, rather than to that closest to $X_0$. \RV{We use SARSOP \cite{Kurniawati2008} and $\gamma = 0.95$ to solve the POMDP and ISC-POMDP (as \eqref{eq:augmentedPOMDP}). Being an offline anytime algorithm, SARSOP had $5$ minutes prior to deployment to compute each policy (and their use online was dominated by belief computation).}

The results of $10,000$ Monte Carlo simulations of the ISC-POMDP and POMDP over $T = 10$ time steps with $X_0 \sim \pi_0$ are summarized in Table \ref{table:results} and Fig.~\ref{fig:results}.
We report the: (average) total discounted cost under the ISC-POMDP cost function in \eqref{eq:iscpomdp} with \eqref{eq:iscCost} (\emph{Discounted Cost}); number of times the agent reaches the correct goal, i.e., the corner closest to its initial position $X_0$ (\emph{No. Goals Reached}); (average) entropy $H(X_0 | y^T, u^{T-1})$ of the final initial-state posterior pmf $p(x_0 | y^T, u^{T-1})$ (\emph{Final Initial-State Entropy}); and, the (average) probability at the (true) initial state $X_0$ in the final (marginal) posterior pmf $p(x_0 | y^T, u^{T-1})$ (\emph{Final Initial-State Prob.}).
Fig.~\ref{fig:results} shows the (average) initial-state entropy $H(X_0 | y^k, u^{k-1})$ and (average) probability at the initial-state in the posterior pmf $p(x_0 | y^k, u^{k-1})$ at other times.
Example realizations of the agent's position are shown in Fig.~\ref{fig:grid2}.

Table \ref{table:results} shows that the ISC-POMDP outperforms the POMDP in terms of the discounted cost and the number of times the agent successfully reaches the corner closest to $X_0$ \RV{(with failures occurring when the measurements do not enable unambiguous estimation of $X_0$)}.
The superior performance of the ISC-POMDP is due to it encoding the agent's exact objective with the initial-state costs \eqref{eq:iscCost} rather than approximating it with the cost \RV{$\kappa(x, u)$}.
Furthermore, the challenge that the ISC-POMDP overcomes (that the POMDP cannot) is that in order for the agent to move to the correct goal, it must first determine its initial state $X_0$.
The lower initial-state entropy and higher posterior probability in Table \ref{table:results} and Figs.~\ref{fig:r1} and \ref{fig:r2} for the ISC-POMDP compared to the POMDP show that the ISC-POMDP selects controls that help to estimate $X_0$, and hence determine the correct goal to move to.
The realizations shown in Fig.~\ref{fig:grid2} illustrate that the ISC-POMDP can take extra steps to estimate the initial state $X_0$ and the correct goal, whilst the POMDP will simply move to the corner closest to its current location $X_k$ when it first becomes confident of its current location.
This experiment illustrates that ISC-POMDPs enable the optimization of costs dependent on an \emph{a priori} unknown initial state $X_0$, which is important since the optimal policy must select controls $U_k$ that resolve uncertainty about the initial state $X_0$, rather than just the current state $X_k$ as in the case of standard ($\rho$-)POMDPs.

\begin{table}[t!]
\centering
\caption{Monte Carlo simulation results. Best values in bold.}
\label{table:results}
\begin{tabular}{@{}ccc@{}}
\toprule
\textbf{Criteria} & \textbf{ISC-POMDP} & \textbf{POMDP} \\ \midrule
\multicolumn{1}{c|}{\textbf{Discounted Cost}}                           & \textbf{6.26}      & 7.91           \\
\multicolumn{1}{c|}{\textbf{No. Goals Reached}}                         & \textbf{8031}      & 4116           \\
\multicolumn{1}{c|}{\textbf{Final Initial-State Entropy}}               & \textbf{1.54}      & 1.72           \\
\multicolumn{1}{c|}{\textbf{Final Initial-State Prob.}}               & \textbf{0.296}      & 0.245           \\ \bottomrule
\end{tabular}
\vspace{-0.4cm}
\end{table}

\begin{figure}[t!]
     \centering
     \subfloat[]{\includegraphics[width=0.85\columnwidth]{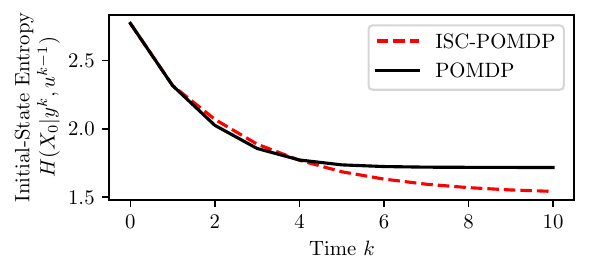}\label{fig:r1}}\\
     \subfloat[]{\includegraphics[width=0.85\columnwidth]{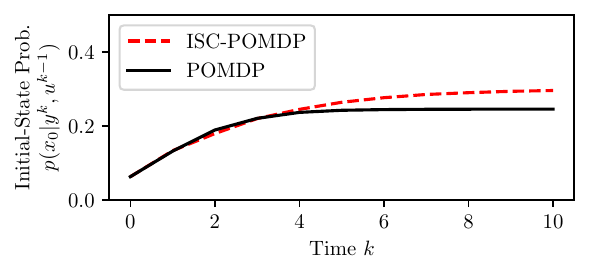}\label{fig:r2}}
    \caption{Simulation Results: (a) Entropy $H(X_0 | y^k, u^{k-1})$ of initial-state posterior pmf $p(x_0 | y^k, u^{k-1})$ . (b) Probability at (true) initial state $X_0$ of posterior pmf $p(x_0 | y^k, u^{k-1})$.}
    \label{fig:results}
\end{figure}

\section{Conclusions and Future Work}
\label{sec:conclusion}
We propose ISC-POMDPs as a class of ($\rho$-)POMDPs with costs dependent on the values and/or uncertainty of initial states.
We use recursive Bayesian smoothing to show that they admit reformulations and solutions as ($\rho$-)POMDPs with augmented states and beliefs.
\RV{Future work will consider problems with continuous state, control, and measurement spaces.}

\bibliographystyle{IEEEtran}
\bibliography{IEEEabrv,Library}

\begin{thebibliography}{10}
\providecommand{\url}[1]{#1}
\csname url@samestyle\endcsname
\providecommand{\newblock}{\relax}
\providecommand{\bibinfo}[2]{#2}
\providecommand{\BIBentrySTDinterwordspacing}{\spaceskip=0pt\relax}
\providecommand{\BIBentryALTinterwordstretchfactor}{4}
\providecommand{\BIBentryALTinterwordspacing}{\spaceskip=\fontdimen2\font plus
\BIBentryALTinterwordstretchfactor\fontdimen3\font minus \fontdimen4\font\relax}
\providecommand{\BIBforeignlanguage}[2]{{%
\expandafter\ifx\csname l@#1\endcsname\relax
\typeout{** WARNING: IEEEtran.bst: No hyphenation pattern has been}%
\typeout{** loaded for the language `#1'. Using the pattern for}%
\typeout{** the default language instead.}%
\else
\language=\csname l@#1\endcsname
\fi
#2}}
\providecommand{\BIBdecl}{\relax}
\BIBdecl

\bibitem{Wang2023}
L.~Wang, I.~R. Manchester, J.~Trumpf, and G.~Shi, ``Differential initial-value privacy and observability of linear dynamical systems,'' \emph{Automatica}, vol. 148, p. 110722, 2023.

\bibitem{Mariottini2011}
G.~L. Mariottini and S.~I. Roumeliotis, ``Active vision-based robot localization and navigation in a visual memory,'' in \emph{IEEE International Conference on Robotics and Automation}, 2011, pp. 6192--6198.

\bibitem{Xue2022}
W.~Xue, R.~Ying, F.~Wen, Y.~Chen, and P.~Liu, ``{Active SLAM With Prior Topo-Metric Graph Starting At Uncertain Position},'' \emph{IEEE Robotics and Automation Letters}, vol.~7, no.~2, pp. 1134--1141, 2022.

\bibitem{Shi2025}
C.~Shi, S.~Han, M.~Dorothy, and J.~Fu, ``Active perception with initial-state uncertainty: A policy gradient method,'' \emph{IEEE Control Systems Letters}, vol.~8, pp. 3147--3152, 2024.

\bibitem{Araya2010}
M.~Araya, O.~Buffet, V.~Thomas, and F.~Charpillet, ``A {POMDP} extension with belief-dependent rewards,'' in \emph{Advances in Neural Information Processing Systems}, vol.~23.\hskip 1em plus 0.5em minus 0.4em\relax Curran Associates, Inc., 2010, pp. 64--72.

\bibitem{Fehr2018}
M.~Fehr, O.~Buffet, V.~Thomas, and J.~Dibangoye, ``{rho-POMDPs have Lipschitz-Continuous epsilon-Optimal Value Functions},'' in \emph{Advances in Neural Information Processing Systems}, vol.~31.\hskip 1em plus 0.5em minus 0.4em\relax Curran Associates, Inc., 2018.

\bibitem{Molloy2023}
T.~L. Molloy and G.~N. Nair, ``{Smoother Entropy for Active State Trajectory Estimation and Obfuscation in POMDPs},'' \emph{IEEE Transactions on Automatic Control}, vol.~68, no.~6, pp. 3557--3572, 2023.

\bibitem{Krishnamurthy2016}
V.~Krishnamurthy, \emph{Partially observed {Markov} decision processes}.\hskip 1em plus 0.5em minus 0.4em\relax Cambridge University Press, 2016.

\bibitem{Cappe2005}
O.~Capp{\'e}, E.~Moulines, and T.~Ryd{\'e}n, \emph{{Inference In Hidden Markov Models}}.\hskip 1em plus 0.5em minus 0.4em\relax New York, NY: Springer, 2005.

\bibitem{Belardinelli2022}
F.~Belardinelli, B.~{G. León}, and V.~Malvone, ``{Enabling Markovian Representations under Imperfect Information},'' in \emph{Proceedings of the 14th International Conference on Agents and Artificial Intelligence}, 2022, pp. 450--457.

\bibitem{Bacchus1996}
F.~Bacchus, C.~Boutilier, and A.~Grove, ``Rewarding behaviors,'' in \emph{Proceedings of the Thirteenth National Conference on Artificial Intelligence}.\hskip 1em plus 0.5em minus 0.4em\relax AAAI Press, 1996, p. 1160–1167.

\bibitem{Bertsekas2005}
D.~P. Bertsekas, \emph{Dynamic programming and optimal control}, 3rd~ed.\hskip 1em plus 0.5em minus 0.4em\relax Belmont, MA: Athena Scientific, 2005, vol.~1.

\bibitem{Kurniawati2008}
H.~Kurniawati, D.~Hsu, and W.~S. Lee, ``{SARSOP}: Efficient point-based {POMDP} planning by approximating optimally reachable belief spaces.'' in \emph{Robotics: Science and Systems}, 2008.

\bibitem{Zheng2023}
W.~Zheng and H.~Lin, ``Provable-correct partitioning approach for continuous-observation {POMDPs} with special observation distributions,'' \emph{IEEE Control Systems Letters}, vol.~7, pp. 1135--1140, 2023.

\bibitem{Demirci2024}
Y.~E. Demirci, A.~D. Kara, and S.~Y\"{u}ksel, ``Average cost optimality of partially observed {MDPs}: Contraction of nonlinear filters and existence of optimal solutions and approximations,'' \emph{SIAM Journal on Control and Optimization}, vol.~62, no.~6, pp. 2859--2883, 2024.

\end{thebibliography}

\end{document}